\documentclass[12pt]{article}
\usepackage{amsmath}
\usepackage{amssymb}
\usepackage{amsthm}
\usepackage{mathrsfs}
\usepackage{setspace}
\usepackage{titling}
\usepackage{esint}
\usepackage[hang,flushmargin]{footmisc}
\makeatletter
\def\th@plain{%
  \thm@notefont{}
  \itshape 
}
\def\th@definition{%
  \thm@notefont{}
  \normalfont 
}
\makeatother
\usepackage{textgreek}
\usepackage{fullpage}
\usepackage{authblk}
\usepackage[pdftex]{graphicx}
\usepackage{titling}
\usepackage{subfigure}
\newtheorem*{theorem}{Theorem}

\begin{document}
\doublespacing
\begin{titlingpage}
\title{A Thermodynamic Structure of Asymptotic Inference}
\author{Willy Wong\thanks{Dept. of Informatics, Faculty of Information Science and Electrical Engineering, Kyushu University \ \texttt{\small willy@inf.kyushu-u.ac.jp}}}
\date{\today}
\maketitle
\thispagestyle{empty}

\begin{abstract}
A thermodynamic framework for asymptotic inference is developed in which sample size and parameter variance define a state space. Within this description, Shannon information plays the role of entropy, and an integrating factor organizes its variation into a first-law–type balance equation. The framework supports a cyclic inequality analogous to a reversed second law, derived for the estimation of the mean. A non-trivial third-law–type result emerges as a lower bound on entropy set by representation noise. Optimal inference paths, global bounds on information gain, and a natural Carnot-like information efficiency follow from this structure, with efficiency fundamentally limited by a noise floor. Finally, de Bruijn’s identity and the I–MMSE relation in the Gaussian-limit case appear as coordinate projections of the same underlying thermodynamic structure. This framework suggests that ensemble physics and inferential physics constitute shadow processes evolving in opposite directions within a unified thermodynamic description.
\end{abstract}
\end{titlingpage}
\section{Introduction}
Statistical inference is the process by which sampled data are used to infer properties of an underlying probability distribution. As the number of samples $m$ becomes large, a familiar asymptotic structure emerges: (i) estimators with asymptotic normality have variances that scale as $1/m$; (ii) independent observations contribute additively; and (iii) sampling distributions become Gaussian. Taken together, these properties can be shown to admit a thermodynamic-like mathematical structure. Although statistical inference itself does not describe thermal interactions, the resulting structure nonetheless exhibits a balance law, a natural third law--type constraint and, under appropriate conditions, a cyclic inequality that closely parallels the fundamental laws of thermodynamics.

These ideas originated from considerations in sensory neuroscience. A sensory receptor does not observe a stimulus parameter directly; instead, it observes noisy microscopic events and must infer the macroscopic stimulus intensity from those events. From a  minimal model of sensory inference, a reverse second-law--type inequality can be derived for cycles of sensory inputs. This correspondence is rigorous and not a metaphor: a restricted form of the inequality has already been tested extensively in neurophysiological recordings. Appendix~1 reviews the relevant sensory background and evidence, while Appendix~2 summarizes the reverse second-law inequality.

This paper takes the next theoretical step. Beginning with the second-law inequality, we construct a thermodynamic state space for inference defined by two coordinates: the sample size $m$ and the observation variance $\sigma^2$ associated with the parameter of interest. The thermodynamic construction is macroscopic, yielding state functions, balance relations, and an effective integrating factor that converts changes in Shannon information into a Clausius-like form. Although this framework emerged from sensory neuroscience, it is not specific to that domain: since the sensory problem is fundamentally that of parameter estimation, the same structure applies---with appropriate qualifications---to metrology, i.e. inference in measurement science.

The inference problem considered here departs from the study of static systems and instead addresses situations where the parameter of interest may change across operating conditions. We assume that inference is conducted independently over non-overlapping epochs, each treated as locally stationary in the statistical sense so that asymptotic estimation theory applies within each epoch. The sample size denotes the number of observations allocated to a single epoch and may vary between epochs to reflect changes in operating conditions. The variance may likewise change, reflecting shifts in the underlying parameter value, distribution, or the heterogeneity of the sampled environment. In this way, non-stationarity is accommodated through shifts in thermodynamic coordinates between epochs.

The present framework shares conceptual ground with several earlier lines of work. Entropy has long been used as an inferential tool, most notably in Jaynes’ maximum-entropy framework \cite{jaynes1957information}. Likewise, information geometry provides a differential-geometric language for Fisher information and statistical manifolds (Amari) \cite{amari2000methods}. In non-equilibrium and stochastic thermodynamics, trajectory-level balance laws and entropy-production inequalities have been developed for physical systems in formulations such as stochastic energetics (Sekimoto) and stochastic thermodynamics (Seifert) \cite{sekimoto1998langevin,seifert2005entropy}. Statistical mechanics has also been used to establish a formal link between estimation error and information theory (e.g. Shental and Kantor) \cite{shental2009shannon}. In these developments, however, the physical quantities retain their standard interpretation tied to microscopic dynamics and energy exchange. In contrast, the framework developed here is a formal analogy: the thermodynamic structure does not describe heat or work, but rather the asymptotic behaviour of inference.

The main contributions of the paper are as follows: (1) a thermodynamic state-space formulation of asymptotic inference in $(m,\sigma^2)$ space, with an explicit uncertainty state function; (2) a reversed second-law–type cyclic inequality for mean inference and an integrating factor that yields a first-law–type balance relation; (3) a third-law–type lower bound on entropy imposed by representation noise, which sets a fundamental noise floor and limits efficiency; (4) optimal inference paths, global bounds on information gain, and a Carnot-like notion of inferential efficiency; and (5) a unification of familiar entropy–information identities, specifically the Gaussian-limit form of the de Bruijn and I–MMSE results, as coordinate projections of a single underlying structure.

With this perspective in place, we begin by contrasting inference with the conventional ensemble-based viewpoint of thermal physics.

\section{Ensemble vs.\ inferential physics}
In thermal physics, the macroscopic description is specified by constraints on variables such as energy or temperature, and physical observables are obtained by averaging over the compatible microscopic configurations. Entropy then quantifies the residual uncertainty about the microstate given the macrostate, and the second law characterizes the typical direction of macroscopic evolution under many independent microscopic interactions. Inference reverses this logic. Consider a tiny probe that has access only to microscopic events, but no direct knowledge of the macroscopic parameters that generate them. Here the macrostate is not prescribed in advance; it is the quantity to be inferred from data. In this sense, statistical inference constitutes the inverse problem to the usual ensemble construction of equilibrium physics.

This reversal is not merely conceptual. The same Gaussian limit that underlies equilibrium thermodynamics also governs asymptotic inference, but it operates in the opposite direction. In thermal physics, many independent randomizing events cause the system to ``forget'' its initial conditions and drive it toward higher entropy. In inference, many independent samples instead sharpen the likelihood, reduce estimator fluctuations and cause the estimator distribution to approach a Gaussian form through the central limit theorem. From here, the asymptotic uncertainty decreases systematically with increasing sample size.

The contrast between the two processes can be summarized succinctly: thermal physics describes the loss of information that emerges from many independent microscopic interactions, whereas inference describes the reduction of uncertainty under repeated sampling. The technical question addressed in this paper is how far this parallel between ensemble physics and inference can be carried. Specifically, we will show that inference in the large sample limit can be endowed with a thermodynamic structure, including state variables, a balance relation and, with additional conditions, a second-law like inequality in $(m,\sigma^2)$ space.

Next, we examine scenarios where the inferential viewpoint is grounded conceptually, namely in sensory transduction and metrology.

\section{Inference in sensory transduction and metrology}
\label{sec:inference}

Sensory transduction provides a concrete and experimentally accessible realization of inferential physics. Sensory receptors interact with stimuli at the periphery through stochastic microscopic events and must estimate macroscopic stimulus parameters from those events. The outcome of this inferential process is accessible experimentally through the neural response. Estimation by the sensory system is achieved through repeated sampling in time. Microscopic events are integrated over successive, approximately independent epochs, and the resulting samples are averaged to form an estimate of stimulus magnitude at each moment. Repetition reduces noise through averaging, and the residual uncertainty of the estimate can be quantified statistically. It is this process of repeated sampling that frames sensory transduction naturally within the realm of statistical inference.

In the large-$m$ limit, the uncertainty in the estimated mean is asymptotically Gaussian with variance $\sigma^2/m$, where $\sigma^2$ denotes the variance of a single observation. In addition, sensory systems are subject to representation noise, which we model as additive, Gaussian and stimulus-independent with variance $\sigma_R^2$. The total uncertainty in the estimate is therefore given by the convolution of these two Gaussian contributions. As a result, the  differential entropy of the estimate takes a closed-form expression depending only on the sample size and the variance \cite{norwich1977information,norwich1993information}.  Note that the calculation of the differential entropy is of the asymptotic estimator distribution, rather than the entropy of the sample vector or the likelihood itself.

Three additional empirical assumptions are required to complete the description. First, the variance of stimulus fluctuations $\sigma^2(\mu)$ is assumed to be a continuous, non-decreasing function of the mean stimulus intensity $\mu$ (c.f. fluctuation scaling \cite{eisler2008fluctuation}). Second, the sample size evolves dynamically, relaxing monotonically toward a unique equilibrium or optimal value $m_{eq}$. Third, since the variance increases monotonically with $\mu$, it is not unreasonable to assume that the optimal sample size increases monotonically with stimulus magnitude $m_{eq}(\mu)$ although not necessarily with the same functional form as $\sigma^2(\mu)$. In what follows, sample size $m$ is treated as a continuous variable, with the approximation justified in the asymptotic regime.

Under these assumptions, a dynamical state-space description follows:
\begin{align}
\dot{m} &= g\!\left(m, m_{eq}\right), \label{diffeqn} \\
H &= \frac{1}{2}\log\!\left(\frac{\sigma^2}{m}+\sigma_R^2 \right) + \text{constant}, \label{entropyH}
\end{align}
where $g$ can be nonlinear and satisfies the conditions that (i) $g(m,m_{eq})=0$ if and only if $m=m_{eq}$, with $m$ tracking $m_{eq}$ monotonically; (ii) $\sigma_R^2$ is constant; and (iii) both $\sigma^2(\mu)$ and $m_{eq}(\mu)$ are continuous, non-decreasing functions of $\mu$. This formulation is deliberately abstract: it does not commit to a specific sensory modality or sampling mechanism, but isolates the generic nature of repeated measurement and averaging \cite{wong2025universal}.

The final step is to relate the uncertainty $H$ to an observable quantity. A direct proportionality between firing rate and uncertainty is assumed where
\begin{equation}
F = k H,
\label{fkh}
\end{equation}
with constant $k$ \cite{norwich1977information,norwich1993information}. This identification is not derived from a deeper microscopic theory; its justification lies entirely in its empirical consequences. Because firing rate is directly measurable, this assumption renders the framework testable experimentally. From this, the resulting \emph{ideal sensory unit} is already sufficient to derive a reversed second-law--type inequality for mean inference (Appendix~2). Additionally, with specific choices of $\sigma^2(\mu)$ and $m_{eq}(\mu)$, the model yields quantitative predictions for firing-rate responses, including a prediction concerning the firing rate that has been found to be obeyed nearly universally in sensory adaptation. The relevant experimental background and evidence are reviewed in Appendix~1.

Sensory systems are not unique in exhibiting this statistical structure. A similar description arises whenever an unknown parameter is inferred from repeated, independent measurements. In metrology, one estimates a parameter $\theta$ of a distribution $p(x\mid\theta)$ from a set of finite samples $x_1,\ldots,x_m$. In the asymptotic regime, estimators with asymptotic normality become Gaussian with variance proportional to $1/m$, independent observations contribute additively, and the estimator entropy admits a closed-form expression identical to (\ref{entropyH}) when representation noise is introduced. The inferential dynamics therefore run in the same direction as in sensory systems: increasing sample size reduces uncertainty. What metrology lacks, without further specification, is a dynamical sampling model (\ref{diffeqn}), and the identification of uncertainty with a directly observable response. Moreover, without the analogue of $F = kH$, the thermodynamic structure describes limits on inferential efficiency rather than experimentally measurable results. Nevertheless, the underlying structure remains the same. In this sense, sensory transduction and measurement theory occupy the same position within a single inferential state space.

Next we turn to formulate a thermodynamics of inference with state variables $(m,\sigma^2)$, entropy $H$, dynamics $\dot m=g(m,m_{eq})$, and trajectory-defined information production.

\section{Foundations for a thermodynamic framework}
The thermodynamic construction rests on two steps: (i) identifying an entropy balance equation, and (ii) ensuring that the entropy is a state function. For metrology, $H$ is a state function if it is a single-valued function of the chosen macroscopic state variables. In the sensory domain, establishing $H$ as a state function is more subtle because the system is driven and adaptive. At equilibrium, however, the sample size relaxes to a unique value $m_{\mathrm{eq}}(\mu)$, yielding $H_{\mathrm{eq}}(\mu)= \tfrac{1}{2}\log\!\left[\sigma_R^2 +\sigma^2(\mu)/m_{\mathrm{eq}}(\mu)\right]$ which depends only on $\mu$; therefore, $H$ is a state function. There is also experimental evidence demonstrating the path-independence of the firing rate $F$ (see discussion in \cite{wong2025universal}).

From here, the differential of the entropy of the asymptotic estimator distribution can be decomposed as
\begin{align}
\label{sensoryentropybalance}
dH &= \frac{\partial H}{\partial \sigma^2}\, d\sigma^2
     + \frac{\partial H}{\partial m}\, dm \\
   &= \delta H_{\mathrm{flux}} + \delta H_{\mathrm{relax}},
\end{align}
where $\delta H_{\mathrm{flux}} = (\partial H/\partial \sigma^2)\,d\sigma^2$ and $\delta H_{\mathrm{relax}} = (\partial H/\partial m)\,dm$ are the coordinate contributions to the exact differential $dH$.  Along relaxation trajectories, we define
\begin{equation}
d\mathscr{I} = -\left(\frac{\partial H}{\partial m}\right) dm,
\end{equation}
or equivalently in rate form $\dot{\mathscr I} = -(\partial H/\partial m)\dot m$. We refer to this quantity as \emph{information production}. Under the conditions discussed in Appendix~2, the theorem there gives the cyclic inequality
\begin{equation}
\label{sensorysecondlaw}
\oint d\mathscr{I} \ge 0.
\end{equation}
This is the second-law--type constraint for inference, with sign reversed relative to thermal physics. In this formulation, a cyclic change in the stimulus (or mean parameter) implies nonnegative net information gain over a cycle.

\section{Identifying an integrating factor and a first law}
Using the explicit form for $H$ in (\ref{entropyH}), we evaluate the differential to be
\begin{equation}
dH= \frac{d\sigma^2}{2 (\sigma^2+m \sigma_R^2)}  - \frac{\sigma ^2dm}{2m (\sigma^2+m \sigma_R^2)} \label{diffbalance}
\end{equation}
where $\sigma^2$ denotes either the variance of a single sensory observation, or in metrological settings the per-observation variance  for estimating parameter $\theta$. The flux contribution to $dH$ is
\begin{equation}
\label{clausiuslike}
dH_{\text{flux}} = \frac{d\sigma^2}{2 (\sigma^2 + m \sigma_R^2)} ,
\end{equation}
and introducing the new state variable
\begin{equation}
\Theta = 2 (\sigma^2 + m \sigma_R^2), \label{Theta}
\end{equation}
changes in the entropy now take the Clausius-like form $dH_{\text{flux}}=\Theta^{-1} d\sigma^2$. Here the state variable $\Theta$ plays a role similar to temperature as an integrating factor. As will be discussed later, at fixed large $m$, $\Theta$ is proportional to the inverse minimum mean square error (MMSE) for an averaged signal corrupted by additive Gaussian noise. An increase in variance produces a larger entropy change at small $m$ than at large $m$. For this reason, we refer to $\Theta$ as the \emph{uncertainty susceptibility}.

It is important to emphasize that in this framework the two coordinates $m$ and $\sigma^2$ are treated independently and interpreted operationally. The quantity $m$ denotes the number of observations assigned to a single inference interval, with intervals taken to be statistically independent and non-overlapping. $\sigma^2$ can change across different epochs, reflecting changes in the underlying parameter value or in the statistical structure of the observations between epochs. Asymptotic relations apply within each interval, with no information carried between intervals. As such, variations in $(m,\sigma^2)$ represents shifts in operating conditions.

Interpreting variance in (\ref{clausiuslike}) as the quasi-heat is natural: variance is positive and quadratic, and often corresponds directly to signal power. Next, we introduce a first-law analogy directly from the state differential (\ref{diffbalance}). Introducing $\Theta$, we get
\begin{equation}
\label{firstlaw}
d\sigma^2 = \Theta \, dH + \frac{\sigma^2}{m}\, dm.
\end{equation}
This is the first law for inferential systems with asymptotically normal estimators. Since $m$ is properly an integer, $dm$ should be seen as the continuum approximation to a small increment $\Delta m$ in the large-$m$ regime (with $\Delta m\ll m$). Since $\sigma^2$ is a state function, a cyclic process satisfies $\oint d\sigma^2 = 0$, and therefore
\begin{equation}
\oint \Theta \, dH = - \oint \frac{\sigma^2}{m}\, dm.
\end{equation}

The first law admits a clear interpretation in terms of uncertainty reduction under repeated sampling. Given sampled data $X_1, X_2, \dots, X_m$ with $\widehat{\theta}(X_1,\dots,X_m)$ an estimator of $\theta$. If the estimator satisfies asymptotic normality such that
\begin{equation}
\widehat{\theta}(X_1,\dots,X_m) \overset{d}{\to} \mathcal{N}\!\left(\theta,\sigma^2/m\right),
\end{equation}
for large $m$, this implies that $\mathrm{Var}(\hat{\theta}_m)  = \sigma^2/m$. Differentiating $\sigma^2 = m\mathrm{Var}(\hat{\theta}_m) $ gives
\begin{align}
d\sigma^2 &= m \; d\left(\frac{\sigma^2}{m}\right) + \frac{\sigma^2}{m} dm.
\end{align}
which, when we reintroduce $dH = m \Theta^{-1} d(\sigma^2/m)$ from (\ref{entropyH}), recovers (\ref{firstlaw}). The first law therefore reflects the $1/m$ scaling of estimator variance together with the additive contribution of independent samples. When variance is introduced into the system, it either contributes to increased uncertainty through the term $\Theta \, dH$, or is reduced through sampling via $(\sigma^2/m)\, dm$. This mirrors the thermodynamic interpretation in which injected energy either increases entropy or is converted into work.

The correspondence with classical thermodynamics can be made explicit by comparing with
\begin{equation}
\label{thermofirstlaw}
dU = T\, dS - P\, dV.
\end{equation}
In thermodynamics, intensive and extensive variables are paired so that $dU$ is extensive (i.e.\ Euler homogeneous of order one under system scaling). In (\ref{firstlaw}), $m$ can be thought of as being extensive since sampling is a resource and additional samples require effort. The variance $\sigma^2$ is also additive under aggregation and is therefore extensive as well. Unlike thermal systems, however, extensivity here depends on how scaling is defined in the inferential setting. In general, $m$ and $\sigma^2$ can be varied independently, and no unique notion of extensivity works across all state variables. If one were to restrict to transformations that preserve the constancy of the estimator variance $\mathrm{Var}(\hat{\theta}_m)=\sigma^2/m$, then increasing $\sigma^2$ requires a proportional increase in $m$.  In such an example, $\Theta = 2(\sigma^2 + m\sigma_R^2)$ scales linearly and is therefore extensive, while $H$, depending only on the ratio $\sigma^2/m$, remains invariant and is thus intensive.

Specific cases further clarify the analogy with thermal physics. For a quasi-isochoric process ($dm = 0$), the quasi-heat equals $d\sigma^2$. In the limit of vanishing representation noise, $\Theta \approx 2\sigma^2$, and the quasi-heat can therefore be written as $d\Theta/2$. For a one-dimensional ideal gas, the corresponding expression is $k_B dT/2$ when normalized to be intensive. Since the variance of the Maxwell–Boltzmann velocity distribution equals $k_B T$, this makes the two sides structurally similar. A gas-law-like relation follows from $\Theta = 2(\sigma^2+m\sigma_R^2)$, yielding $\sigma^2/m = \Theta/(2m) - \sigma_R^2$. In the limit of zero representation noise, this reduces to $(\sigma^2/m) m \approx \Theta/2$.  Finally, in the case where sampled values are positively correlated (due to sampling too quickly within a correlation time period, for example), an additional positive term is required to correct for the variance.  In the thermodynamic case, this would be analogous to having particles which have a \textit{repulsive} force, thus adding to the internal energy.

The first law outlines a number of processes similar to thermal physics.  We can define paths under constant uncertainty $H$ (quasi-adiabatic) such that $\sigma^2/m$ is constant.  In this case, $dH=0$ and $d\sigma^2=(\sigma^2/m)\, dm \propto dm$.  Another is where $\Theta=2(\sigma^2+m\sigma_R^2)$ is constrained to be constant (quasi-isothermal), for which $d\sigma^2=-\sigma_R^2 \, dm$.  Finally, in the case where the representation noise can be ignored, $\Theta=2\sigma^2$ and thus $dH=-(1/2m)\, dm$ for which $H$ takes on a logarithmic form.

The quasi-work or \emph{sampling work} term $(\sigma^2/m)\, dm$ can be interpreted as the incremental variance expenditure required to increase the effective sample size by $dm$. It has units of variance and reflects both the precision of the data and the work associated with changing the sample size. When variance is constant, this term can be integrated between $m_a$ and $m_b$ to give $\sigma^2\log(m_b/m_a)$.  If $m_a=m_b$, the sampling work is zero.  Sampling work has three main characteristics. First, it illustrates diminishing returns with increasing sample size. Second, larger $\sigma^2$ (noisier data) requires more work to achieve the same improvement in precision. Third, it shows scale invariance, in that the work depends on the ratio $m_b/m_a$ rather than on the absolute difference between the sample sizes. In the general case, the integral of sampling work is path dependent since $(\sigma^2/m)\,dm$ is not an exact differential, and its value non-zero when evaluated over a cycle. To illustrate this, consider a system with variance $\sigma_a^2$ sampled with sample size $m_a$. At $t=0$, the system is perturbed so that the variance increases to $\sigma_b^2 \ge \sigma_a^2$ while the sample size is increased to $m_b$. At a later time, the variance returns to $\sigma_a^2$ and the sample size returns to $m_a$. The sampling work over this cycle is
\begin{align}
\oint \frac{\sigma^2}{m} dm
&= \int_{m_a}^{m_b} \frac{\sigma_b^2}{m}\,dm
+ \int_{m_b}^{m_a} \frac{\sigma_a^2}{m}\,dm \\
&= (\sigma_b^2-\sigma_a^2)\log\!\left(\frac{m_b}{m_a}\right),
\end{align}
which can be interpreted geometrically as the weighted area enclosed in $(m,\sigma^2)$ space.

\section{Other thermodynamic relationships, and a third law}
Other thermodynamic potentials can be written down formally, but they do not have the same operational role as in thermal physics. The reason is that the variance susceptibility $\Theta=2(\sigma^2+m\sigma_R^2)$ is a derived state function which is fixed jointly by both $m$ and $\sigma^2$.  So the usual Legendre formalism of fixing one variable and extremizing over the conjugate variable does not have a direct analogue.  For example, one may define a quasi–Helmholtz free energy
\begin{equation}
\mathscr{A} = \sigma^2 - \Theta H .
\end{equation}
This is a legitimate state function but because $\Theta$ is not an independently controllable variable, $\mathscr{A}$ does not generate a distinct thermodynamic potential with an associated extremum principle.  Moreover, the corresponding Gibbs-like potential reduces to $\mathscr{G}=-\Theta H$. Finally, an enthalpy-like construction collapses to zero because constant $\sigma^2/m$ implies $dH=0$.  In this case, all of the variance must be accounted for by the sampling work.

Several partial derivative relations can be derived.  These relations follow directly from the entropy function $H$ and $\Theta$:
\begin{equation}
\begin{aligned}
\label{partialdiff}
\left(\frac{\partial H}{\partial \sigma^2}\right)_m
&= \frac{1}{\Theta} , \quad &
\left(\frac{\partial H}{\partial m}\right)_{\sigma^2}
&= -\frac{\sigma^2}{m \Theta}, \\
\left(\frac{\partial \sigma^2}{\partial m}\right)_H
&= \frac{\sigma^2}{m}, \quad &
\left(\frac{\partial \Theta}{\partial m}\right)_H
&= \frac{\Theta}{m}. \\
\end{aligned}
\end{equation}
From the equation of state $\Theta=2(\sigma^2+m\sigma_R^2)$ one also obtains the quasi-specific heat $C_m=1/2$, reminiscent of an equipartition-like result.

Finally, a natural third law also emerges from the thermodynamic structure. The entropy $H = \tfrac{1}{2}\log\!\left(\sigma^2/m+\sigma_R^2 \right) + \text{const}$ attains its minimum value as $m\to\infty$, where (up to a constant) it approaches the value $\log(\sigma_R)$. In sensory applications one often chooses the constant so that $H=0$ as $m\to\infty$, allowing $H$ to be interpreted as mutual information. As in thermodynamics, this limit is unattainable: not only is $m$ a resource, in the sensory model $m$ assumed to be finite and saturates at $m_{eq}$. As such, this illustrates the unattainability of zero entropy.

\section{Information gain and efficiency}
An important concept in thermodynamics is the efficiency by which energy supplied to a system can be converted into useful work. The analogous question in inference concerns the efficiency with which information can be acquired through sampling. In asymptotic estimation theory, estimator variance decreases with sample size as $1/m$ under asymptotic normality; the Cram\'er--Rao bound provides the corresponding efficient limit when regularity conditions hold. Fisher's information, however, is not a measure of uncertainty in the literal sense; it characterizes the curvature of the log-likelihood. By contrast, Shannon entropy provides an axiomatic measure of uncertainty, but is typically defined for fixed probability distributions and does not by itself encode how uncertainty changes with sampling. In (\ref{entropyH}), $H$ couples Fisher with Shannon information through an equation of uncertainty that is explicitly parameterized by sample size, allowing uncertainty reduction and inferential efficiency to be treated within a single state-space framework.  Next we explore how this function shapes information gain and efficiency.

\subsection{Information gain}
We begin by considering the information gained given a fixed sampling work budget.  This problem can be formulated by examining the sampling-driven terms from the first and second laws.  Consider the extremization of the information gained through the variation of 
\begin{equation}
\Delta \mathscr{I}= -\displaystyle\int_{m_a}^{m_b} \left(\frac{\partial H}{\partial m}\right)_{\sigma^2}\,dm
= \displaystyle\int_{m_a}^{m_b} \frac{1}{\Theta}\frac{\sigma^2}{m}\,dm,
\end{equation}
which we maximize subject to a fixed amount of sampling work
\begin{equation}
\int_{m_a}^{m_b} \frac{\sigma^2}{m} dm=\mathscr{W}.  
\end{equation}
The resulting Euler-Lagrange equation is given by
\begin{equation}
\frac{\partial}{\partial \sigma^2} \left(\frac{\sigma^2}{2m\left(\sigma^2+m\sigma_R^2 \right)}-\frac{\lambda \sigma^2}{m} \right)=0
\end{equation}
where $\lambda$ is the Lagrange multiplier.  The condition of optimality can be written in the compact form $\Theta/\sqrt{m}=\text{const}$ illustrating a trade-off between the state variable $\Theta$ and resource $m$ at a fixed sampling work budget. 

Solving for the optimal variance trajectory we obtain
\begin{equation}
\label{optimalvar}
\sigma_{\text{optimal}}^2(m)=\frac{\mathscr{W}+\sigma_R^2(m_b-m_a)}{2\left(\sqrt{m_b}-\sqrt{m_a}\right)}\sqrt{m}-m\sigma_R^2
\end{equation}
This solution can be shown to yield a maximum for the information gained by considering the second variation.  The optimal trajectory has an inverted u-shape due to the tradeoff between a sub-linear term $\sqrt{m}$ and linear term $m$.  Trajectories are nevertheless restricted by the non-negativity of variance.

From (\ref{optimalvar}), the optimal information gained can be evaluated to be
\begin{equation}
\Delta \mathscr{I}_{\text{optimal}}=\frac{1}{2}\log \frac{m_b}{m_a}-\frac{2\sigma_R^2\left(\sqrt{m_b}-\sqrt{m_a}\right)^2}{\mathscr{W}+(m_b-m_a)\sigma_R^2}
\end{equation}
Any other path will yield lower information gain for the same budget $\mathscr{W}$. The optimal value of $\Delta \mathscr{I}$ is non-decreasing in both $m_b/m_a$ and $\mathscr{W}$. From here, an upper bound can be established:
\begin{equation}
\label{maxinfo}
\Delta \mathscr{I}_{\text{max}}=\frac{1}{2}\log \frac{m_b}{m_a}.
\end{equation}
Regardless of how sampling takes place, or how the variance evolves, the difference in uncertainty going from $m_a$ to $m_b$ cannot exceed this upper bound. Moreover, since the derivation is carried out in the asymptotic limit, the bound is independent of the underlying distribution, provided that variance exists in the asymptotic regime. While this result may appear straightforward when variance is held constant, \emph{it is no longer immediate when variance is allowed to vary}. Thus, (\ref{maxinfo}) defines the maximum information gain for asymptotic inference across all admissible paths in $(m,\sigma^2)$ space --- a channel-capacity–like upper bound.

\subsection{Information efficiency}
Next we consider the efficiency by which information is generated or produced.
\subsubsection{Identities linking estimation and information theory}
Changes in uncertainty can be summarized by the two partial derivative relationships
\begin{gather}
\left(\frac{\partial H}{\partial \sigma^2}\right)_m= \frac{1}{\Theta} \label{dh/ds}\\
\left(\frac{\partial H}{\partial m}\right)_{\sigma^2}= -\frac{\sigma^2}{m \Theta} \label{dh/dm}
\end{gather}
Equation (\ref{dh/dm}) describes the change in uncertainty with varying sample size while variance is  fixed. Writing the effective noise level as $t=\sigma^2/m$, increasing $m$ corresponds to a decrease in $t$, i.e.\ a diffusion process in reverse time. Let $X \sim \mathcal N(0,\sigma_R^2)$, $Z\sim \mathcal N(0,1)$, and define $Y_t = X + \sqrt{t}\,Z$. The differential entropy $H(Y_t)$ then evolves according to the Gaussian form of de~Bruijn’s identity
\begin{equation}
\frac{d}{dt} H(Y_t)=\frac{1}{2}I(Y_t),
\end{equation}
where $I(Y_t)$ is the Fisher information. Equation (\ref{dh/dm}) is therefore equivalent to de~Bruijn’s identity under the change of variables $t=\sigma^2/m$, and describes entropy evolution under a reverse diffusion process \cite{cover1999elements}.

Equation (\ref{dh/ds}), on the other hand, specifies how entropy responds to changes in variance at fixed sample size. In communication theory, SNR measures channel quality while in inference, SNR instead quantifies residual uncertainty relative to inferential effort. Although the same ratio appears, its interpretation differs. For an additive Gaussian channel with signal-to-noise ratio $\mathrm{SNR}=\sigma^2/m\sigma_R^2$, the minimum mean-square error satisfies
\begin{equation}
\mathrm{MMSE}=\frac{2\sigma^2 \sigma_R^2}{\Theta}.
\end{equation}
Thus changes in entropy follows the Gaussian form of the I--MMSE relation \cite{guo2005mutual}
\begin{equation}
\frac{d}{d\,\mathrm{SNR}} I(X;Y)=\frac{1}{2}\,\mathrm{MMSE}.
\end{equation}
where $I(X;Y)$ is the mutual information. The appearance of both de~Bruijn’s identity and the I--MMSE relation reflects the Gaussian structure underlying the entropy function. More importantly, they arise as complementary coordinate projections of a single underlying thermodynamic structure. This unification is encoded by the state function $\Theta = 2(\sigma^2 + m\sigma_R^2)$, which governs joint sampling- and variance-driven information transduction.

\subsubsection{Carnot-like efficiency}
For fixed $m$, the quantity $\Theta$ is bounded from below. Its conditional minimum occurs at $\sigma^2=0$, where $\Theta_C = 2m\sigma_R^2$. The representation noise $\sigma_R^2$ therefore imposes a fundamental lower bound in $\Theta$ analogous to a third-law--type constraint. Next, we rewrite the MMSE as
\begin{equation}
\label{inferencecarnot}
\eta \triangleq \frac{\mathrm{MMSE}}{\sigma^2/m}
= \frac{\Theta_C}{\Theta},
\end{equation}
where normalization by the input signal variance $\sigma^2/m$ defines a dimensionless quantity $\eta$ characterizing local information efficiency. The efficiency obeys $0 \le \eta \le 1$, with $\eta=1$ attained trivially when $\sigma^2=0$, or when $\sigma^2 \ne 0$ and $m \to \infty$ which is impossible since $m$ is a resource. The existence of a temperature-like variable $\Theta$ together with a bounded efficiency permits Carnot-type reasoning at the level of inference. Here $\Theta_C$ plays the formal role of a cold-reservoir temperature, limited by the representation noise. Unlike thermal engines, where efficiency measures work extraction and therefore involves a one minus structure, the efficiency here measures the fraction of attainable certainty achieved.

While the preceding discussion concerns estimators that achieve asymptotic normality, the definition of information efficiency extends naturally to \emph{efficient estimators}. The Cram\'er--Rao lower bound gives $\mathrm{Var}(\hat{\theta}_m) \ge 1/m I(\theta)$, where $I(\theta)$ is the Fisher information associated with a single observation. From the uncertainty susceptibility, we have
\begin{equation}
\Theta=2\left(\sigma^2 + m\sigma_R^2\right) \ge 2\left(I(\theta)^{-1} + m\sigma_R^2\right).
\end{equation}
Relative to the baseline, estimator inefficiency appears as an excess contribution to the variance. Estimators that achieve the lower bound therefore have the highest information efficiency $\eta$.  As such, efficient estimators are directly analogous to Carnot engines in thermal physics.

\subsubsection{Principle of least entropy production}
An alternative perspective of (\ref{dh/ds}) can be gained by examing entropy production.  Define entropy production as
\begin{equation}
d\Sigma= \frac{1}{\Theta}\,d\sigma^2,
\end{equation}
where $d\Sigma$ denotes the path dependent entropy production associated with an incremental displacement in $(m,\sigma^2)$ space.  For a fixed $d\sigma^2$, we seek the condition that renders $d\Sigma$ stationary, i.e.
\begin{equation}
\delta(d\Sigma)=0,
\end{equation}
for which $\delta(d\Sigma)=\delta(\Theta^{-1})\, d\sigma^2$. If $m$ is varied alone, no stationary point exists because $\partial(\Theta^{-1})/\partial m \neq 0$ for finite $m$.  Stationarity therefore requires simultaneous variation of both coordinates such that $d\Theta=0$, i.e. $\Theta$ remains constant along the stationary path.  The second variation is positive, so this condition corresponds to a minimum.  Thus, small changes in variance produce the least possible entropy production when $\Theta$ is constant.

The requirement that $\Theta=2(\sigma^2+m\sigma_R^2)$ remain constant can appear counterintuitive, as an increase in variance must be accompanied by a decrease in sample size.  This reflects the fact that information efficiency decreases with both increasing variance and increasing sample size.  When $\Theta$ is held fixed, the marginal cost of reducing uncertainty remains uniform.  Thus the principle of local minimal information production is therefore not a statement about information itself, but about the cost associated with changes in uncertainty or information.

\subsubsection{Efficiency of information cyclic}
The measure of local efficiency can also be shown to constrain \emph{global information efficiency}. To see this, consider the ratio $\mathscr{I}/\mathscr{W}$, where $\mathscr{I}=\oint (\sigma^2/m\Theta)\, dm$ denotes the total information gained over a cycle and $\mathscr{W}=\oint (\sigma^2/m)\, dm$ represents the corresponding sampling work. This ratio can be interpreted as a weighted average of $\Theta^{-1}$ with respect to the measure $d\mathscr{W}=(\sigma^2/m)\, dm$ along the cycle. Consequently, it must lie between the minimum and maximum values of $\Theta^{-1}$ attained along the path. Moreover, assume that $\Theta^{-1}$ attains its supremum along the cycle at $\Theta_{\inf}^{-1}$. It follows that $\mathscr{I}/\mathscr{W}=\langle \Theta^{-1} \rangle \le \Theta_{\inf}^{-1}$. From the definition $\Theta = 2(\sigma^2 + m\sigma_R^2)$, it follows that $\Theta_{\inf} \ge 2 m_* \sigma_R^2$, where $m_*$ is the minimum value of $m$ attained along the cycle. Therefore $\mathscr{I}/\mathscr{W} \le \Theta_{C}(m_*)^{-1}$. Thus the global information efficiency is bounded ultimately by the local efficiency evaluated at the point where the sample size reaches its minimum along the cycle. In other words, the representation noise $\sigma_R^2$ sets the irreducible floor, while the smallest attainable sample size determines where this floor is realized. Global efficiency is limited by the extreme values of $\Theta$ encountered along the path. In particular, the upper bound depends on access to the smallest achievable sample size. This mirrors Carnot's result in thermodynamics, where the efficiency of a cycle depends only on the availability of extreme temperatures.

\section{Discussion}
This paper introduces a thermodynamic framework to asymptotic inference by adopting sample size $m$ and variance $\sigma^2$ as macroscopic state variables. Within this description, the differential entropy of the asymptotic estimator distribution is calculated in (\ref{entropyH}), and relationships formally analogous to the laws of thermodynamics can be derived. These include a reversed second-law--type inequality for mean inference (\ref{sensorysecondlaw}), the existence of temperature-like state function $\Theta = 2(\sigma^2 + m\sigma_R^2)$ leading to a first-law--type balance relation (\ref{firstlaw}), and a lower bound on uncertainty set by representation noise that plays a role analogous to a third law. Taken together, these elements define a coherent state-space structure governing information acquisition in the asymptotic regime permitting the derivation of fundamental results in information gain and efficiency.

The correspondence developed here is not accidental. The same thermodynamic mathematical structure—state functions, balance relations, and inequality constraints—appears in both equilibrium thermal physics and asymptotic inference because both rely on additivity and central-limit-type aggregation of many independent contributions. What differs is the directionality. In ensemble physics, macroscopic constraints are imposed and microscopic degrees of freedom are averaged over, leading to entropy increase under coarse-graining. In inference, microscopic observations are accumulated to infer macroscopic parameters, leading to entropy decrease under repeated sampling. In this sense, inferential dynamics becomes a reversed counterpart of thermal physics, with a reversed second law, an exchange of intensive and extensive roles, and an interchange between those quantities that are state functions and those that are path dependent. The thermodynamic structure therefore does not encode a particular physical mechanism, but instead constrains both the forward ensemble approach and its inverse inferential process as two realizations of a common formal framework.

When extending this approach to metrology, one would expect the same picture and mathematics to apply, but with several important caveats. First, the theory is asymptotic: it assumes large sample size $m$, asymptotic normality of estimators, and sufficient smoothness to treat $m$ as a continuous variable. These assumptions are not always satisfied in practical settings of measurement. Second, the derivation of a second-law--type inequality relies on some minimum admissibility constrains on sampling, as well as a fluctuation-scaling asymmetry. In sensory systems fluctuation-scaling is well supported empirically, as larger stimulus magnitudes are typically accompanied by larger fluctuations. In metrology, such scaling is not automatic, although it is often reasonable to expect that larger quantities are associated with larger errors or fluctuations. Finally, and most importantly, metrology lacks a natural analogue of the empirical relation $F = kH$ linking uncertainty to a directly observable quantity. As a result, the thermodynamic quantities defined here retain theoretical meaning but do not correspond directly to measurable observables in the same way they do in the sensory domain.

The theoretical framework described here concerns systems that evolve, while assuming that inference proceeds across statistically independent, non-overlapping epochs. Each epoch is treated as locally stationary, so that classical estimation theory applies within each epoch with no information carried forward. This provides a tractable description of inference in a changing system, although it is necessarily an idealization. If, in practice, epochs are not fully independent, the main consequence will be the inflation of the observed variance. For example, drift in the underlying mean can introduce additional between-epoch variability, so that the recorded dispersion exceeds the nominal value of $\sigma^2/m$ when such effects are not modelled explicitly. From the standpoint of the present theory, such deviations—provided they can be represented as additional stochastic variability and do not disrupt the asymptotic properties—will not alter the geometric structure of the $(m,\sigma^2)$ state space, but instead shift the operating point of the system to that of a lower efficiency.

Several directions for future work follow naturally from this framework. While the present theory focuses on a thermodynamic description of inference, the identification of variance as an energy-like variable also suggests a similar \emph{statistical physics} viewpoint. If the individual samples are denoted by $x_1,\dots,x_m$, a microstate may be identified with the vector $\mathbf{x}\in\mathbb{R}^m$. Fixing the sample mean $\bar{x}$ and sample variance $\hat{\sigma}^2 = \tfrac{1}{m}\sum_i (x_i-\bar{x})^2$ restricts the microstates to a sphere of fixed radius, analogous to a kinetic energy shell, and thereby defining a microcanonical ensemble. Alternatively, maximizing entropy subject to a constraint on $\langle\sigma^2\rangle$ yields a canonical distribution $P(\mathbf{x}) \propto \exp\!\left(-\beta \sum_i (x_i-\bar{x})^2\right)$, with $\beta = \Theta^{-1}$. Moreover, if the sample entropy $\hat{H}$ is constructed from the sample variance $\hat{\sigma}^2$, this formulation allows for considerations of fluctuations in inference. In particular, large-deviation results such as Sanov’s theorem suggest a route towards the construction of fluctuation theorems for inference \cite{cover1999elements}. Finally, although the present work is formulated in the large-$m$ limit, systematic corrections can be explored for specific sampling distributions, such as exponential-family, using the Edgeworth or related asymptotic expansions.
\appendix
\section{Appendix 1: Sensory background}
This appendix summarizes the empirical and theoretical background of the sensory work that motivated the present paper. The goal is not provide a full account of the entire theory, but to show how the \emph{ideal sensory unit} (cf Section~\ref{sec:inference}) reproduces results and predictions reported previously.  Full details of the theory can be found elsewhere, e.g. \cite{wong2025universal}.

\subsection{Linearization of sampling dynamics}
A standard way to solve the nonlinear state-space model (\ref{diffeqn})-(\ref{entropyH}) is to linearize it around a fixed point. Here the fixed point is the steady-state sample size $m_{eq}$. For constant stimulus, expand $\dot{m}=g(m,m_{eq})$ to first order about $m=m_{eq}$ and use $g(m_{eq},m_{eq})=0$ to obtain
\begin{equation}
\label{relax}
\frac{dm}{dt} = -a\, (m - m_{eq}),
\end{equation}
where $a=-(\partial g/\partial m)\rvert_{m=m_{eq}}>0$ is the relaxation rate.

\subsection{A Tweedie scaling ansatz and an expression for entropy}
To obtain closed-form expressions for the response of the sensory neuron (i.e. the firing rate), one can adopt a fluctuation-scaling-type model for the input statistics. The canonical example arises in vision, where photon statistics at the level of the photoreceptor are well approximated by a Poisson distribution. In particular, if the input distribution lies in the Tweedie family, then over a wide range of natural signals one can write $\sigma^2 \propto \mu^p$, and correspondingly $m_{eq}\propto \mu^{p/2}$.

Choosing the additive constant in (\ref{entropyH}) to be $\tfrac{1}{2}\log(\sigma_R^2)$ renders $H$ to be a mutual information. Writing the mean input as $\mu = I+\delta I$, with stimulus intensity $I$ and additive noise $\delta I$, we have
\begin{gather}
\label{net}
H = \frac{1}{2} \log \left( 1 + \frac{\beta (I + \delta I)^p}{m} \right), \\
m_{eq} = (I + \delta I)^{p/2},
\end{gather}
where the parameters $k$, $\beta$, $p$, $\delta I$ and $a$ are taken to be positive. In practice, these parameters are estimated by fitting to response data, although overfitting is often a challenge. These exact equations have been developed across a series of publications beginning in the 1970's \cite{norwich1977information,norwich1993information,norwich1995universal,wong1997physics,wong2020rate}. Given a stimulus profile $I(t)$, one can compute the time-dependent firing rate via $F=kH$ in (\ref{fkh}).

\subsection{Sensory adaptation}
Adaptation refers to the neural response to stimulation whereby firing rate rises from a spontaneous level, reaches a peak shortly after stimulus onset, and then decays monotonically to a steady-state value. For a constant stimulus $t\ge 0$, $m_{eq}$ is constant and (\ref{relax}) has the solution
\begin{equation}
\label{msolution}
m(t) = m(0)e^{-at} + m_{eq} (1 - e^{-at}).
\end{equation}
By continuity we take $m(0)$ to be the equilibrium sample size at $I=0$, i.e. $m(0)=\delta I^{p/2}$ \cite{wong2020rate}. Substituting (\ref{msolution}) into (\ref{net}) and using $F=kH$ we obtain the time-dependent response
\begin{equation}
\label{adaptation}
F(I, t) = \frac{k}{2} \log \left( 1 + \frac{\beta (I + \delta I)^p}{ \delta I^{p/2} e^{-at} + (I + \delta I)^{p/2} (1 - e^{-at}) } \right).
\end{equation}

This expression has three natural fixed points typically observed in adaptation responses: the spontaneous rate $\mathrm{SR}=F(0,\infty)$, the peak rate $\mathrm{PR}=F(I,0)$, and the steady-state rate $\mathrm{SS}=F(I,\infty)$,
\begin{gather}
\label{SRPRSS}
\mathrm{SR} = \frac{k}{2} \log \left( 1 + \beta \delta I^{p/2} \right), \\
\mathrm{PR} = \frac{k}{2} \log \left( 1 + \frac{\beta (I + \delta I)^p}{ \delta I^{p/2} } \right), \\
\mathrm{SS} = \frac{k}{2} \log \left( 1 + \beta (I + \delta I)^{p/2} \right).
\end{gather}

\subsection{A universal inequality for sensory adaptation}
The three points, together with the curvature of the logarithm, can be used to derive an inequality governing the steady-state activity:
\begin{equation}
\label{inequality}
\sqrt{\mathrm{PR}\times \mathrm{SR}} \le \mathrm{SS} \le \frac{\mathrm{PR}+\mathrm{SR}}{2}.
\end{equation}
That is, the steady-state response must lie between the geometric and arithmetic means of the spontaneous and peak activities \cite{wong2023fundamental}. The simplicity of (\ref{inequality}) belies its reach: a constraint that is obeyed almost universally across different sensory modalities despite differences in mechanisms and stimulation paradigms.

The inequality was evaluated across two comprehensive analyses encompassing 40 separate studies and more than 400 individual recordings of adaptation from different sensory modalities, animal species, and laboratories. The first analysis examined auditory adaptation data from the 1970's to the present. Hearing is particularly well suited for repeated adaptation experiments because auditory stimuli can be precisely controlled and single-unit responses readily isolated. Across datasets, $\mathrm{SS}$ was plotted against $\mathrm{PR}$ and compared with the predicted bounds in (\ref{inequality}). No parameter fitting was required: $\mathrm{SR}$ was obtained directly from the baseline activity. Figure~\ref{fig1}, reproduced from \cite{wong2023fundamental}, shows empirical trajectories superimposed on the theoretical bounds.

\begin{figure*}[!h]
\begin{center}
\includegraphics[width=.89\textwidth]{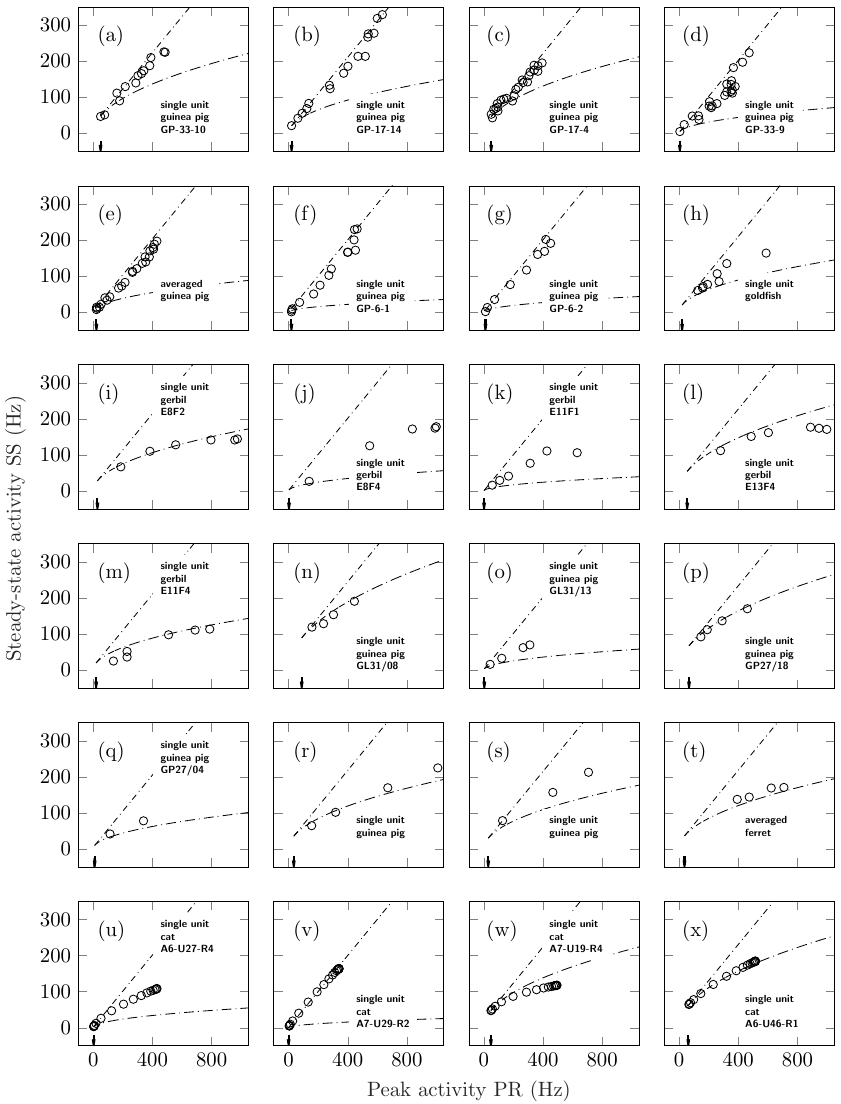}
\caption{Figure reproduced from \cite{wong2023fundamental}. Steady-state activity plotted versus peak activity for a number of auditory studies.  In all panels, the dashed lines show the theoretical upper and lower bounds of (\ref{inequality}).  No fitted parameters were required to plot these bounds.  SS vs PR from (a-g) single or averaged guinea pig fibre recordings (figs. 11a-d, 12, 17a and 4a from \cite{smith1975short}); (h) saccular nerve fibres of goldfish (fig. 3 from \cite{fay1978coding}); (i-m) single fibre gerbil recordings (figs. 4 and 5 from \cite{westerman1984rapid}); (n-q) single guinea pig fibre recordings (figs. 1 and 2 from \cite{yates1985very}); (r-s) single guinea pig fibre recordings (figs 3a and 3b from \cite{muller1991relationship}); (t) averaged ferret data (fig. 6 from \cite{sumner2012auditory}); (u-x) single cat fibre recordings (figs. 12e-h from \cite{peterson2021simplified}).  The spontaneous activity of each unit is indicated by an arrow pointing towards the x-axis.}\label{fig1}
\end{center}
\end{figure*}

In a second study, the analysis was extended across other modalities \cite{wong2021consilience}. Since spontaneous activity is not always reported, one can use the lower-bound approximation $\mathrm{SS} \approx \sqrt{\mathrm{PR}\times\mathrm{SR}}$. On a log--log plot, this gives the prediction $\mathrm{SS} \propto \mathrm{PR}^{1/2}$. Across eight major sensory modalities (proprioception, touch, taste, hearing, vision, smell, electroreception, and temperature) and four animal phyla (Chordata, Arthropoda, Mollusca, and Cnidaria), most datasets exhibited slopes near $1/2$ with high correlation.

The convergence of these findings implies a form of universality. Neurons of different types do not share identical mechanisms, yet the same quantitative relation holds across modalities, species, and measurement paradigms. For example, data recorded nearly a century ago by Adrian and Zottermann \cite{adrian1926impulsesb} show the same adherence as modern recordings, underscoring robustness across time and methodology.

\section{Appendix 2: Proof of the information inequality}
We summarize the proof of the information inequality in (\ref{sensorysecondlaw}).  Recall that the information differential is defined by
\begin{equation}
d\mathscr{I} = -\left(\frac{\partial H}{\partial m}\right) d m,
\end{equation}
where $H$ is the entropy state function in (\ref{entropyH}) and $dm$ is evaluated along trajectories generated by $\dot m=g(m,m_{eq})$.  Since $\sigma^2$ is a monotonic function of the stimulus parameter $\mu$, we will refer to the state space using coordinates $(\mu,m)$ for the purposes of analyzing stimulus–response cycles.

\begin{theorem}[Information inequality]
Let $\mu(t)$ be a cyclic stimulus that generates a closed stimulus--response trajectory in $(\mu,m)$ via the dynamical state-space model described in Section~\ref{sec:inference}.
The change in information over a cycle obeys the second-law--type inequality
\begin{equation}
\oint d \mathscr{I} \ge 0.
\end{equation}
\end{theorem}

\begin{proof}
Let $C$ denote the closed curve traced in $(\mu,m)$ space over one stimulus cycle. For any piecewise-smooth simple loop $C$ that is positively oriented (counter-clockwise) and bounds a region $A$, Green's theorem gives
\begin{equation}
\label{green}
\oint_C d\mathscr{I}
= -\oint_C \frac{\partial H}{\partial m}\,dm
= -\iint_A \frac{\partial^2 H}{\partial m\,\partial \mu}\,dA.
\end{equation}
For the sensory uncertainty state function (\ref{entropyH}), we evaluate the mixed derivative to yield
\begin{equation}
\label{mixedderivative}
\frac{\partial^2 H}{\partial m\,\partial \mu}
= -\frac{\sigma_R^2}{2m^2\left( \sigma_R^2 + \sigma^2/m \right)^2}\,\frac{d\sigma^2}{d\mu}.
\end{equation}
Under monotonicity of $\sigma^2$ and $\mu$, $d\sigma^2/d\mu\ge 0$, the mixed derivative is non-positive throughout $A$.
Therefore the area integral in (\ref{green}) is non-negative, implying $\oint_C d\mathscr{I}\ge 0$ for any positively oriented loop.

To justify the orientation for realizable cycles, the sampling dynamics always relaxes toward $m_{eq}(\mu)$: $m$ tracks the equilibrium $m_{eq}(\mu)$ with no overshoot such that increases in $\mu$ drive up $m$, while decreases in $\mu$ drive down $m$. Consequently, $dm/d\mu \ge 0$ along the increasing phase of the stimulus and $dm/d\mu \le 0$ along the decreasing phase, enforcing a counterclockwise traversal of the loop in $(\mu,m)$ space. Hence, physically realizable sensory cycles are positively oriented, and the inequality follows.
\end{proof}

The second-law inequality, in a more restricted form, has already been tested extensively in neurophysiological recordings. A cyclic presentation of a stimulus, in which a sensory neuron is initially at rest, then stimulated with a constant input (i.e.\ adaptation) until a steady state is reached, and finally returned to rest, constitutes a cyclic process compatible with the second law. The sum of differences in the firing rates over the cycle can be calculated from $-\oint (\partial H / \partial m)\, dm$ and $F = kH$, yielding $(\mathrm{PR}-\mathrm{SS}) + (\mathrm{TR}-\mathrm{SR})$, where $\mathrm{PR}$, $\mathrm{SS}$, $\mathrm{TR}$, and $\mathrm{SR}$ denote the peak, steady-state, trough, and spontaneous firing rates respectively. Since sensory neurons are observed to obey the inequality~(\ref{inequality}), as well as the associated inverted inequality $\sqrt{\mathrm{TR} \times \mathrm{SS}} \le \mathrm{SR} \le (\mathrm{TR} + \mathrm{SS})/2$, a straightforward calculation shows that $(\mathrm{PR}-\mathrm{SS}) + (\mathrm{TR}-\mathrm{SR}) \ge 0$ as is required for the second-law \cite{wong2025universal}. Moreover, this exact result can also be obtained directly from the property of \emph{supermodularity} imposed by the mixed derivative in (\ref{mixedderivative}).

In metrological settings, sampling strategies need not be governed by the same dynamical rules as in sensory systems; however, they are still subject to minimal admissibility constraints. These include a monotonic relationship in which larger parameter magnitudes are associated with larger intrinsic errors or variance, as well as an adaptive response whereby increased uncertainty drives increased sampling work. In addition, changes in sampling strategy do not occur instantaneously, but relax toward an optimal operating point over finite time. Under these weak and generic conditions, closed cycles are constrained to be positively oriented in $(\mu,m)$ space with a non-negative mixed derivative, similar to the sensory case. Consequently, the cyclic inequality continues to hold, extending the second-law–type constraint to metrological inference in estimating the mean.
 \bibliographystyle{plain}
\bibliography{thermodynamics}
\end{document}